\documentclass[11pt]{article}
\usepackage{vmargin}
\setmarginsrb{1in}{1in}{1in}{1in}{0cm}{0.0cm}{0cm}{0.6cm}

\newcommand{\squishlist}{
  \begin{list}{$\bullet$}
    { \setlength{\itemsep}{0pt}      \setlength{\parsep}{1pt}
      \setlength{\topsep}{0pt}       \setlength{\partopsep}{0pt}
     \setlength{\leftmargin}{.6em} \setlength{\labelwidth}{1.5em}
      \setlength{\labelsep}{0.5em} } }
\newcommand{\squishend}{
   \end{list}  }

\usepackage{amsmath}
\usepackage{amssymb}
\usepackage{amstext}
\usepackage{epsfig}

\usepackage{theorem}
\usepackage[algoruled,vlined,algosection]{algorithm2e}
\SetKwInOut{Input}{Input}
\SetKwInOut{Output}{Output}
\newtheorem{dfn}{Definition}[section] 
\newtheorem{prop}[dfn]{Proposition}
\newtheorem{theorem}[dfn]{Theorem}

\newtheorem{rem}[dfn]{Remark}
\newtheorem{lemma}[dfn]{Lemma}
\newtheorem{claim}[dfn]{Claim}

\newcommand{\proof}{\noindent \emph{Proof.} }

\newcommand{\qed}{\hspace*{\fill}$\square$}
\newcommand{\remove}[1]{}

\newcommand{\bw}{\operatorname{bw}}

\newcommand{\mids}{\operatorname{mid}}

  \renewcommand{\int}{\operatorname{int}}       

\begin{document}

\title{Planar Subgraph Isomorphism Revisited 
}
\author{
Frederic Dorn \\
\footnotesize Department of Informatics, University of Bergen, Norway\\
\small {\tt frederic.dorn@ii.uib.no}
}

\date{}
\maketitle

\begin{abstract}
The problem of {\sc Subgraph Isomorphism} is defined as follows:
 Given a \emph{pattern} $H$  and a  \emph{host graph} $G$ on $n$ vertices, does $G$ contain a subgraph that is isomorphic to $H$? Eppstein [SODA 95, J'GAA 99] gives the first linear time algorithm for  subgraph isomorphism for a fixed-size pattern, say of order $k$, and arbitrary planar host graph, improving upon the $O(n^{\sqrt{k}})$-time algorithm when using the ``Color-coding'' technique of Alon et al [J'ACM 95]. 
Eppstein's algorithm runs in time $k^{O(k)} n$, that is,
the dependency on $k$ is superexponential. 
We solve an open problem posed in Eppstein's paper and 
improve the running time to $2^{O(k)} n$, that is, single exponential in $k$ while keeping the term in $n$ linear. 
 Next to deciding subgraph isomorphism, we can construct a solution and count all solutions in the same asymptotic running time. We may enumerate $\omega$ subgraphs with an additive term $O(\omega k)$ in the running time of our algorithm.
We introduce the technique of ``embedded dynamic programming" on a suitably structured graph decomposition, which exploits the topology of the underlying  embeddings of the subgraph pattern (rather than of the host graph).
 To achieve our results, we give an upper bound on the number of partial solutions in each dynamic programming step  as a function of pattern size---as it turns out, for the planar subgraph isomorphism
problem, that function is single exponential in the number of vertices in the pattern. 
 \end{abstract}

\section{Introduction}\label{intro}

  \noindent In the literature, we often find results on polynomial time or even linear time algorithms for NP-hard problems. 
  Take for example the NP-complete problem of computing an optimal tree-decomposition of a graph. Bodlaender~\cite{Bodlaender96} gives an algorithm in time $O(n)$ for this problem---restricted to input graphs of constant treewidth. The  Graph Minor Theory developed by Robertson and Seymour implies amongst others that there is an $O(n^3)$ algorithm for the disjoint path problem, that is for finding disjoint paths between a constant number of terminals. 
   Taking a closer look at such results, one notices that 
    a function exponential in size of some constant $c$ is hidden in the $O$-notation of the running time---here, $c$ is the treewidth and the number of terminals, respectively. 
   In another line of research, \emph{parameterized complexity}, the primary goal is to rather find algorithms that minimize the exponential term of the running time.
   The first step here is  to prove that such an algorithm with a separate exponential function exists, that is,
   that the studied problem is  \emph{fixed parameter tractable (FPT)}~\cite{DowneyF99,FlumGrohebook,Niedermeierbook06}. Such problem has  an algorithm with time complexity  bounded by a function of the form $f(k) \cdot n^{O(1)}$, where the {\em parameter function} $f$ is a computable function only depending on $k$.  
The second step in the design of FPT-algorithms is to decrease  the growth rate of the parameter function.

 We can identify two different trends in which  running times of  exact algorithms are improved. 
 First, one can decrease the degree of the polynomial term in  the asymptotic running time,  and second,
   one can focus on  obtaining parameter functions with better exponential growth.
 In the present work, we achieve both goals for the computational problem {\sc Planar Subgraph Isomorphism}.


 
 {\sc Subgraph Isomorphism} generalizes  many important graph problems, such as {\sc Hamiltonicity}, {\sc Longest Path}, and {\sc Clique}. 
It is  known to be $NP$-complete, even when restricted to planar graphs
~\cite{GareyJ79}. Until now, the best known algorithm  to solve {\sc Subgraph Isomorphism}, that is to find a subgraph of a given host graph isomorphic to a pattern $H$ of order $k$ (the number of vertices in $H$), is the na\"{i}ve exhaustive search algorithm with running time $O(n^{k})$ and no FPT-algorithm can be expected here~\cite{DowneyF99}.
 For a pattern $H$ of treewidth at most $t$, Alon et al.~\cite{AlonYZ95} give an algorithm of running time $2^{O(k)}  n^{O(t)}$. 
 For {\sc Planar Subgraph Isomorphism}, given planar pattern and input graph, some  considerable improvements have been made mostly during the 90's. The first improvement was provided by Plehn and Voigt~\cite{PlehnV90}, with running time $O(k^{k})  n^{O(\sqrt{k})}$. Using  the elegant Color-coding technique of Alon et al.~\cite{AlonYZ95}, one can devise an algorithm of running time  $2^{O(k)}  n^{O(\sqrt{k})}$. The current benchmark has been set by Eppstein~\cite{Eppstein99} to $k^{O(k)}  n$, by employing graph decomposition methods, similar to the Baker-approach~\cite{Baker94} for approximating NP-complete problems on planar graphs. Eppstein's algorithm is actually the first FPT-algorithm for {\sc Planar Subgraph Isomorphism} with $k$ as parameter. Eppstein poses three open problems: a) whether one can extend the technique in~\cite{AlonYZ95} to improve the dependence on the size of the pattern from $k^{O(k)}$ to $2^{O(k)}$  for the decision problem of subgraph isomorphism; and whether one can achieve similar improvements b) for the counting version
and c) for the listing version of the subgraph isomorphism problem.

\paragraph{\bf Our results.}
In this work, we do not only achieve this single exponential behavior  in $k$ for all three problems---without applying the randomized coloring technique---we also  keep
  the term in $n$ linear. That is, 
we 
give an algorithm for {\sc Planar Subgraph Isomorphism} for a  pattern $H$ of order $k$ with running time $2^{O(k)}  n$. 
 Next to deciding subgraph isomorphism, we can construct a solution and count all solutions in the same asymptotic running time. We may list $\omega$ subgraphs with an additive term $O(\omega k)$ in the running time of our algorithm.
Our algorithm also improves the time complexity of the previous approach~\cite{FominT06} for patterns of size $k \in o(\sqrt{n}\log n)$.

    The novelty of our result  comes from \emph{embedded dynamic programming}, a technique we find interesting on its own.
 Here, one decomposes the graph by separating it into induced subgraphs.  In the dynamic programming step, one computes partial solutions for the separated subgraphs, that are updated to an overall solution for the whole graph.
In ordinary dynamic programming, one would  argue how the subgraph pattern hits  separators of the host graph.
 Instead, in embedded dynamic programming for subgraph isomorphism,  we proceed exactly the opposite way: we look at how separators can be routed through the subgraph pattern. As a consequence, we  bound the number of partial solutions not by a function of the separator size of the host graph, but by a function of the pattern size---as it turns out, for the planar subgraph isomorphism
problem, that function is single exponential in the number of vertices of the pattern.
To obtain a good bound on the parameter function, we apply several fundamental enumerative combinatorics results in the technical sections of this work. Next to the number of cycles and face-vertex sequences in embedded graphs, 
 these counting results give upper bounds on the number of planar triangulations 
  and planar embeddings of the  pattern.

 

 Our algorithm is divided into two parts with the second part being the aforementioned embedded dynamic programming. For keeping the time complexity of our algorithm linear in the size of the host graph, we give a fast method  for computing
  a graph decomposition with special properties:  \emph{Sphere-cut decompositions}
    are natural extensions of tree-decompositions to plane graphs, where the separator vertices are connected by a Jordan curve.  In  embedded dynamic programming we use sphere-cut decompositions with separators of size linearly bounded by the size of the subgraph pattern. 
   
 
 \begin{theorem}
\label{bigtheorem}
Let $G$ be  a planar graph  on $n$ vertices and ¨$H$ a pattern of order $k$.
We can decide if there is a subgraph of $G$ that is isomorphic to $H$ in time $2^{O(k)}  n$.
 We find subgraphs and count  subgraphs of $G$ isomorphic to $H$ in time $2^{O(k)}  n$ and enumerate $\omega$ subgraphs in time $2^{O(k)}n+O(\omega k)$.
\end{theorem}

 It is worth mentioning that for $k$-{\sc Longest Path} on planar graphs, the authors of~\cite{DornPBF09}  give the first algorithm with time complexity  subexponential in the parameter value. The algorithm has running time $2^{O(\sqrt{k})} n + O(n^{3})$, employing the techniques \emph{Bidimensionality} and topology-exploiting dynamic programming. 
\emph{Bidimensionality Theory} employs results of  Graph Minor Theory by Robertson and Seymour  for planar graphs~\cite{RobertsonST94} and other structural graph classes to algorithmic graph theory (entry~\cite{DemaineFHT05}, for a survey~\cite{DemaineH07-II}). Unfortunately, Bidimensionality does only work for finding specific patterns in a graph, such as $k$-paths, but not for  subgraph isomorphism problems in general.
For a survey on other planar subgraph isomorphism problems with restricted patterns, 
 please consider~\cite{Eppstein99}.

\paragraph{\bf Organization.} After giving some definitions in Section~\ref{sect:prelim}, we show in Section~\ref{sect:compsc} how to obtain a sphere-cut decomposition of small width. In Section~\ref{sec:planesubiso}
 we restrict {\sc Planar Subgraph Isomorphism} to {\sc Plane Subgraph Isomorphism}.
  We first give some technical lemmas in Section~\ref{subsec:combnoose} to bound the number of ways a separator of the sphere-cut decomposition can be routed through a plane pattern. We describe and analyze  embedded dynamic programming in Section~\ref{subsec:embdp} followed by subsuming the entire algorithm for {\sc Plane Subgraph Isomorphism} in Section~\ref{subsec:alg}.
  In Section~\ref{sec:planesubiso} we bound the number of drawings of the pattern  and show how to solve {\sc Planar Subgraph Isomorphism}.


\section{Preliminaries}
\label{sect:prelim}

\vspace{-2mm}\paragraph{\bf  Subgraph isomorphism.} Let $G,H$ be two graphs. We call $G$ and $H$ \emph{isomorphic} if there exists a bijection $\nu : V(G) \rightarrow V(H)$ with
 $\{v,w\} \in E(G) \Leftrightarrow \{\nu(v),\nu(w)\} \in E(H)$. We call $H$ \emph{subgraph isomorphic to $G$} if there is a subgraph $H'$ of $G$ isomorphic to $H$.

\vspace{-2mm}\paragraph{\bf Branch Decompositions.} A {\em branch decomposition}~$\langle T,\mu \rangle$ of a graph
$G$ consists of an unrooted ternary tree $T$ (i.e., all internal vertices have 
degree three) and a bijection $\mu: L \rightarrow E(G)$
from the set $L$ of leaves of $T$ to the edge set of $G$. We
define for every edge~$e$ of~$T$ the {\em middle set}~$\mids(e)
\subseteq V(G)$ as follows: Let $T_1$ and $T_2$ be the two
connected components of $T\setminus \{e\}$. Then let $G_i$ be the
graph induced by the edge~set~$\{\mu(f): f \in L \cap V(T_i) \}$
for $i \in \{1,2\}$. The \emph{middle set} is the intersection of
the vertex sets of $G_1$ and $G_2$, i.e., $\mids(e):= V(G_1) \cap
V(G_2)$. The {\em width} $\bw$ of $\langle T,\mu \rangle$ is the
maximum order of the middle sets over all edges of $T$, i.e.,
$\bw(\langle T,\mu \rangle) := \max\{|\mids(e)| \colon e\in T\}$.
An optimal branch decomposition of $G$ is defined by a tree $T$
and a bijection $\mu$ which together provide the minimum width,
the {\em branchwidth} $\bw(G)$.

\vspace{-2mm}\paragraph{\bf Plane graphs and equivalent embeddings.} 
Let $\Sigma$ be the unit sphere.
A \emph{plane drawing} or \emph{planar embedding}  of a graph $G$ with vertex set $V(G)$ and edge set $E(G)$ maps
vertices  to points in the sphere, and edges to
simple curves between their end vertices, such that edges do not cross,
except in common end vertices. A  \emph{plane graph} is a graph G together with
a plane drawing. A \emph{planar graph} is a graph that admits a plane drawing. For details, see e.g.~\cite{Diestel00}.
The set of \emph{faces} $F(G)$ of a plane graph $G$ is defined as the union of the connected regions of $\Sigma \setminus G$.
A subgraph of a plane graph $G$, induced by the vertices and edges incident to a face $f \in F(G)$, is called a \emph{bound} of $f$.
  If $G$ is 2-connected, each bound of a face is a cycle. We call this cycle \emph{face-cycle} (for further reading, see e.g.~\cite{Diestel00}). 
 For a subgraph $H$ of a plane graph $G$, we refer to the drawing of $G$ reduced to the vertices and edges of $H$ as a~\emph{subdrawing} of $G$.
Consider any two drawings $G_1$ and $G_2$ of a planar graph $G$.
 A \emph{homeomorphism} of $G_1$ onto $G_2$ is a homeomorphism of $\Sigma$ onto itself which maps vertices, edges, and faces of $G_1$ onto
 vertices, edges, and faces of $G_2$, respectively.
We call two planar drawings of the same graph \emph{equivalent}, if they are homeomorphic. 
 
\begin{theorem}(e.g.~\cite{Diestel00})
\label{prop:whitney} 
Every 3-connected planar graph has a unique  embedding in a sphere $\Sigma$ up to homeomorphism. 

\end{theorem} 
 
\vspace{-2mm}
\paragraph{\bf Triangulations.}
We call a plane graph $G$ a \emph{planar triangulation} or simply a \emph{triangulation} if every face in $F(G)$ is bounded by a triangle (a cycle of length three). 
 If $H$ is a subdrawing of a triangulation $G$, we call $G$ a \emph{triangulation of $H$}.


\vspace{-2mm}\paragraph{\bf Nooses and combinatorial nooses.} A \emph{noose} of a $\Sigma$-{plane} graph $G$ is
 a simple closed curve in $\Sigma$ that meets $G$ only in vertices.
  From the Jordan Curve Theorem, it then follows that nooses
separate $\Sigma$ into two regions. 
Let $V(N) = N \cap V(G)$ be the vertices  and $F(N)$ be 
 the  faces intersected by a noose $N$. 
 The \emph{length}  of $N$ is  the number $|V(N)|$ of vertices in $V(N)$. The clockwise order in which $N$ meets the vertices of $V(N)$ is a cyclic permutation $\pi$ on the set $V(N)$.
\begin{rem}
\label{radialcycle_NEW}
Let  a plane graph $H$ be a subdrawing of a plane graph $G$. 
Every noose $N$ in $G$ is also a noose  in $H$ and $V_H(N) \subseteq V_G(N)$.
\end{rem}

\noindent A \emph{combinatorial noose} $N_C = [v_0,f_0,v_1,f_1,\ldots,f_{k-1},v_k]$  in a plane graph $G$ is 
 an alternating sequence of vertices and faces of $G$,   such  that 
\squishlist 
\item $f_i$ is a face incident to both $v_i,v_{i+1}$ for all $i < k$,
\item $v_0=v_k$ and the vertices $v_1,\ldots,v_k$ are mutually distinct and 
\item if $f_i = f_j$ for any $i\neq j$ and $i,j = 0,\ldots,k-1$, 
         then the vertices $v_{i},v_{i+1},v_j$, and $v_{j+1}$  do not appear in the order $(v_i,v_j,v_{i+1},v_{j+1})$ on the bound of face $f_i=f_j$.
\squishend
 The \emph{length} of a combinatorial noose $[v_0,f_0,v_1,f_1,\ldots,f_{k-1},v_k]$ is $k$.

\begin{rem}
\label{obs:comb_noose}
The order in which  a noose $N$ intersects the faces $F(N)$ and the vertices $V(N)$ of a plane graph $G$  gives a unique alternating face-vertex sequence of $F(N) \cup V(N)$ which is a combinatorial noose $N_C$. Conversely, for every combinatorial noose $N_C$ there exists a noose $N$ with face-vertex sequence $N_C$.
\end{rem}

\noindent We may view combinatorial nooses as equivalence classes of nooses, that can be represented by the same face-vertex-sequence.

\vspace{-2mm}\paragraph{\bf Sphere cut decompositions.} 
For a~$\Sigma$-plane graph~$G$, we define a~\emph{sphere cut
 decomposition} or~\emph{sc-decomposition}~$\langle
T,\mu ,\pi \rangle$ as a branch decomposition which for every
edge~$e$ of~$T$  
has a  noose~$N_e$ that cuts $\Sigma$ into two regions~$\Delta_1$ and~$\Delta_2$ such that~$G_i \subseteq \Delta_i
\cup N_e$, where  $G_i$ is the
graph induced by the edge~set~$\{\mu(f): f \in L \cap V(T_i) \}$ for $i \in \{1,2\}$ and $T_1 \dot{\cup} T_2 = T\setminus \{e\}$. Thus $N_e$ meets $G$ only in
$V(N_e)=\mids(e)$ and its length is $|\mids(e)|$. 
  The vertices of every middle
set~$\mids(e) = V(G_1) \cap V(G_2)$ are enumerated according to a 
 cyclic permutation $\pi$ on $\mids(e)$.




\vspace{-2mm}\paragraph{}The following two propositions will be crucial in that they give us upper bounds on the number of partial solutions we will compute in our dynamic programming approach. With both propositions, we will bound the number of combinatorial nooses in a plane graph by the number of cycles in the triangulation of some auxiliary graph.
 With the second proposition we  bound the number of non-equivalent embeddings of planar graphs.  

\begin{prop}(\cite{BuchinKKSS07})
\label{totallength}
No planar $n$-vertex graph has more than $2^{1.53n}$ simple cycles. 
\end{prop}

\begin{prop} (\cite{Tutte62})
\label{lem:tutte}
 The number  of non-isomorphic maximal planar  graphs on $n$ vertices is approximately $2^{3.24n}$. 
\end{prop}
\noindent Proposition~\ref{lem:tutte} also gives a bound on the number of non-isomorphic triangulations.
Any embedding of a maximal planar graph $G$ must be a
triangulation, otherwise $G$ would not be maximal.  
With Theorem~\ref{prop:whitney}, every maximal planar graph has a unique embedding which is a triangulation.
On the other hand, every triangulated graph is maximal planar. 

\section{Computing sphere-cut decompositions in linear time}
\label{sect:compsc}
\noindent 
In this section we introduce an algorithm for computing sc-decompositions of bounded width.  
Let $H$ be a connected subgraph of $G$ with $|V(H)|= k$, and let $v \in V(H)$. Then $H$ is a subgraph of the induced subgraph $G^v$ of $G$, where $G^v=G[S]$ with $S =\{ w \in S  \mid $ dist$(v,w) \leq k\}$ (dist$(v,w)$ denotes the length of a shortest path between $v$ and $w$ in $G$). This observation helps us to shrink the search space of our algorithm by cutting out chunks of $G$ of bounded width  and solve subgraph isomorphism separately on each chunk.
 With the algorithm of Tamaki~\cite{Tamaki03}, one can compute a branch decomposition of $G^v$ of width $\leq 2k+1$, following similar ideas as in the approach of Baker~\cite{Baker94} for tree decompositions.
 With some simple modifications, we achieve the same result for sc-decompositions. 
In Appendix~\ref{app:tamaki_alg} we prove the following lemma and give an algorithm that computes  a sc-decomposition of bounded width in linear time.

\begin{lemma}
\label{lem:tamaki}
Let $G$ be a plane graph with  a rooted spanning tree whose root-leaf-paths have length at most $k$. We can find an sc-decomposition of width $2 k +1 $
   in time $O(k n)$.
\end{lemma}




\section{Plane subgraph isomorphism}
\label{sec:planesubiso}

\noindent In this section, we study the subgraph isomorphism problem on patterns and host graphs that 
 are embedded in a sphere $\Sigma$. 
 In Section~\ref{sec:planarsubiso} we carry over our results to planar graphs. 
 We first introduce  some topological tools that allow us to define a refined 
  dynamic programming approach.
    At  every step of the dynamic programming approach, we  compute all possibilities of how a combinatorial noose $N$
corresponding to a middle set of the sc-decomposition $\langle T,\mu, \pi \rangle$ of $G$ can intersect a subdrawing equivalent  to pattern $H$. Each intersection gives rise to a combinatorial noose of $H$. See Figure~\ref{fig:dp1} for an illustration.
\begin{figure}[h]
\begin{center}
\includegraphics[width=0.4\textwidth]{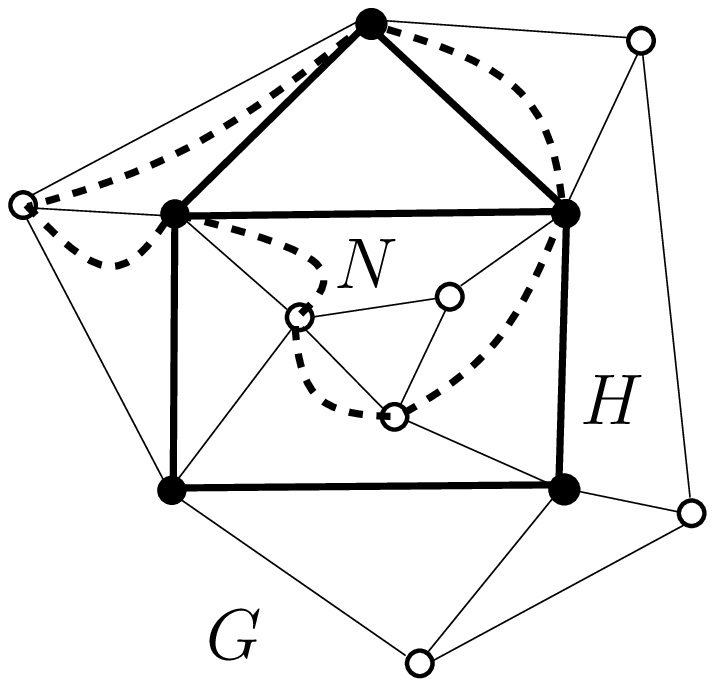}
 \includegraphics[width=0.4\textwidth]{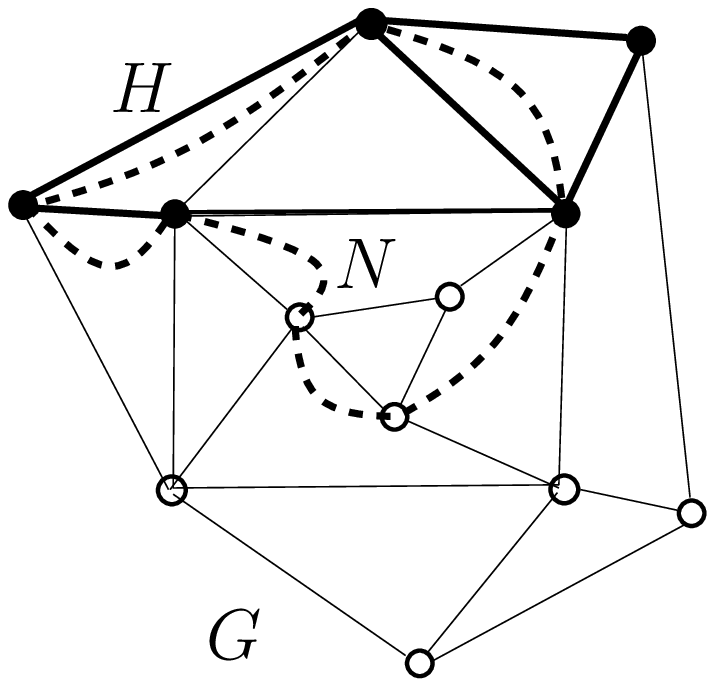}
\end{center}
\caption{{\small On the left, we have a plane graph $G$ with an emphasized subdrawing $H$  intersected by a combinatorial noose $N$  indicated by dashed lines. 
  On the right, we have the same graph $G$ with a different copy of $H$  intersected by $N$.}} 
\label{fig:dp1}
\end{figure}  
 
 The running time of the algorithm  crucially depends on the number of combinatorial nooses in $H$.
   The aim of this section is to prove the following:

 \begin{theorem}
\label{theor:planeSI}
 Let $G$ be a plane graph on $n$ vertices and $H$ be a plane graph on $k \leq n$ vertices.
  We can decide if there is a subdrawing of $G$ that is equivalent to $H$ in time $2^{O(k)}n$.
   We can find and count subdrawings equivalent to $H$ in time  $2^{O(k)}n$, and enumerate $\omega$ subdrawings
    in time $2^{O(k)}n+O(\omega k)$. 
 \end{theorem}

\vspace{-2mm}\subsection{Combinatorial nooses in plane graphs}
\label{subsec:combnoose}


\noindent For a refined algorithm analysis we now take a close look at combinatorial nooses of plane graphs.
 In particular we are interested in counting the number of combinatorial nooses.  
 In this subsection, we will prove the following lemma:
 
 \begin{lemma}
\label{lem:numradwalks}
Every plane $k$-vertex graph has $2^{O(k)}$ 
 combinatorial nooses.
\end{lemma}

 Before proving this lemma,  
 we  show that  every combinatorial noose of a plane graph on $k$ vertices corresponds to a cycle in some other plane graph on  at most $O(k)$ vertices. 
    First we relate combinatorial nooses in a planar triangulation $H'$ to the cycles in $H'$. In a second step we relate combinatorial nooses of a 3-connected plane graph $H$ to cycles in the triangulations of $H$.
  Finally, we will show that for any plane graph $H$ there is an auxiliary graph $H^*$, such that the combinatorial nooses of  $H$ can be injectively mapped to the cycles of the triangulations of $H^*$.  From Proposition~\ref{totallength} we know an upper bound on the number of cycles in planar graphs, which we employ to prove  Lemma~\ref{lem:numradwalks}.

\begin{lemma}
 \label{lem:radialwalks_NEW}
 Let $H$ be  a planar triangulation and $N_C= [v_0,f_0,v_1,f_1,\ldots,f_{k-1},v_k]$ 
 a combinatorial noose of $H$. 
 Then for every pair of consecutive vertices $v_i,v_{i+1}$ in $N_C$, there is an edge ${v_i,v_{i+1}}$ in $E(H)$. 
  That is, the sequence $[v_0,v_1,\ldots,v_k]$ is a simple cycle in $H$ if $|V(N_C)|>2$, and if $|V(N_C)|=2$, it corresponds to a single edge in $H$.
  \end{lemma}
 \proof 
Since $H$ is triangulated, we have that every $f_i \in N_C$ is bounded by a triangle $\Delta$ where $v_i,v_{i+1}$ are two of the three vertices of $\Delta$ and $v_i,v_{i+1}$ have an edge in common.
 Since  vertices occur only once in $N_C$, $f_i$ is unique in $N_C$ if $|V(N_C)|>2$, that is, there is no $f_j \in N_C$ with $i \neq j$ and $f_i = f_j$. Hence we map each $f_i$ one-to-one to edge $e_i=\{v_i,v_{i+1}\}$ and get a cycle  $[v_0,e_0,v_1,e_1,\ldots,e_{k-1},v_k]$. For $|V(N_C)|=2$, $f_0$ and $f_1$ are  incident faces to edge $\{v_0,v_1\}$. 
 For an illustration, see Figure~\ref{fig:noosecycle}. 
 \begin{figure}[h]
\begin{center}
\includegraphics[width=0.33\textwidth]{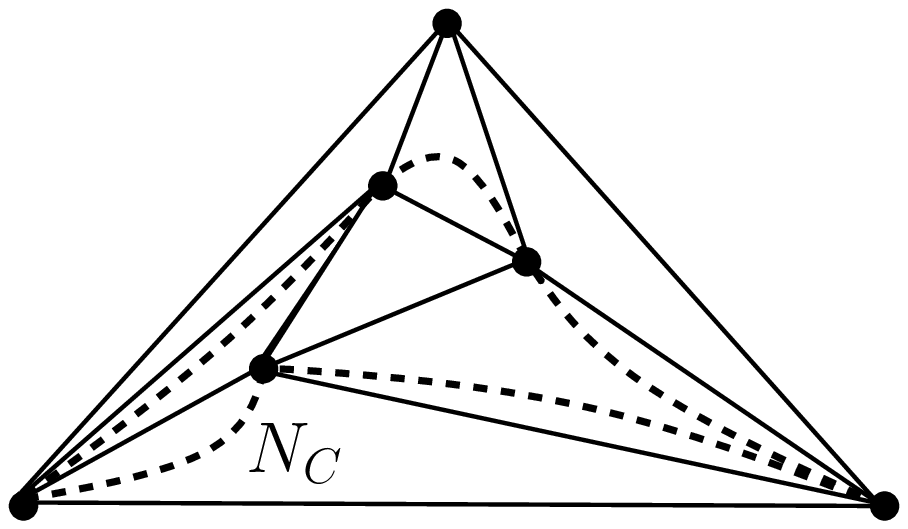}
 \includegraphics[width=0.33\textwidth]{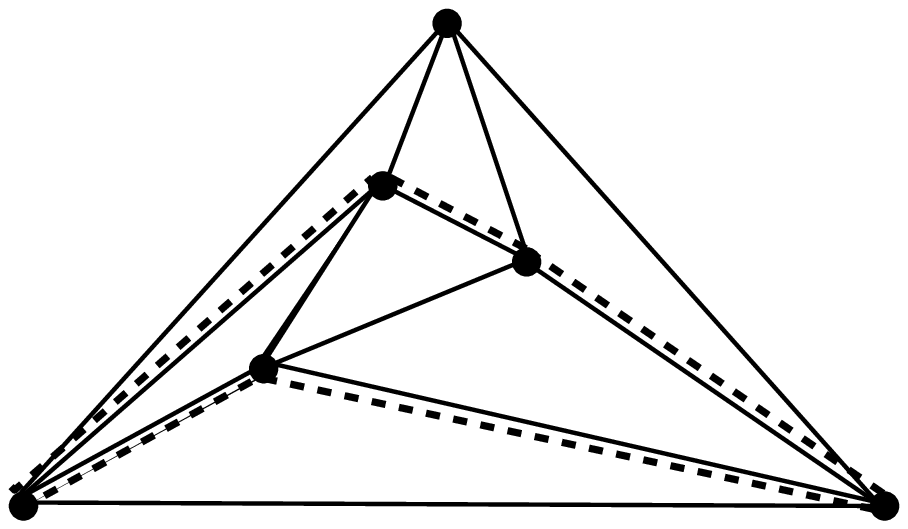}
\end{center}
\caption{{\small On the left, there is a triangulation with a combinatorial noose $N_C$ indicated by dashed lines. On the right, we have mapped the noose to a cycle indicated by dashed lines.}} 
\label{fig:noosecycle}
\end{figure}
\qed

\begin{lemma}
\label{lem:walktriang_NEW}
Let $H$ be a 3-connected plane  graph and $N_C= [v_0,f_0,v_1,f_1,\ldots,f_{k-1},v_k]$ 
 a combinatorial noose of $H$ with $|V(N_C)|>2$.
  Then there exists a planar triangulation $H'$ of $H$, such that 
   $[v_0,v_1,\ldots,v_k]$ is a cycle in $H'$.   
\end{lemma}
\proof 
We proceed in two phases. First we iteratively add edges to $H$ and transform $N_C$ into another combinatorial noose such that every two consecutive vertices in $N_C$ have a common edge. Then we triangulate the resulting graph.

 For every pair of consecutive vertices $v_i,v_{i+1}$ in $N_C$, if $v_i,v_{i+1}$ have no edge in common, add $e_i=\{v_i,v_{i+1}\}$  to $E(H)$.
 Thereby the drawing of $e_i$ splits $f_i$ into two new faces $f^a_i$ and $f^b_i$, bounded by face-cycle $C^a$ and $C^b$  respectively, where $C^a \cap C^b = e_i$.
 Since $N_C$ corresponds to a noose by Remark~\ref{obs:comb_noose} and nooses are not self-intersecting, we observe the following for $|V(N_C)|>2$: for every $f_j=f_i$ in $N_C$ with $j \neq i$ we have that  both $v_j,v_{j+1}$ are in one of $C_a$ and $C_b$. Thus, adding edge $e_i$ will not cross any other edge added in this process.
  In $N_C$, we replace $f_i$ by one of $f^a_i$ and $f^b_i$, and every $f_j = f_i$ by $f^a_i$ if $F_j$ is bounded by $C^a$ and by $f^b_i$ otherwise. 
 Once we have an edge  for every pair of consecutive vertices  in $N_C$, we  note that for every sub-sequence $[v_i,f_i,v_{i+1}]$ of $N_C$ the edge $e_i=\{v_i,v_{i+1}\}$ is incident to face $f_i$ since, by 3-connectivity, edge $e_i$ is uniquely embedded in $H$.
 We then add edges arbitrarily to obtain a triangulation $H'$. By Lemma~\ref{lem:radialwalks_NEW}, the vertices of $N_C$ correspond to a cycle in $H'$. 
 For an illustration, see Figure~\ref{fig:noosecycle2}. 
 \begin{figure}[h]
\begin{center}
\includegraphics[width=0.33\textwidth]{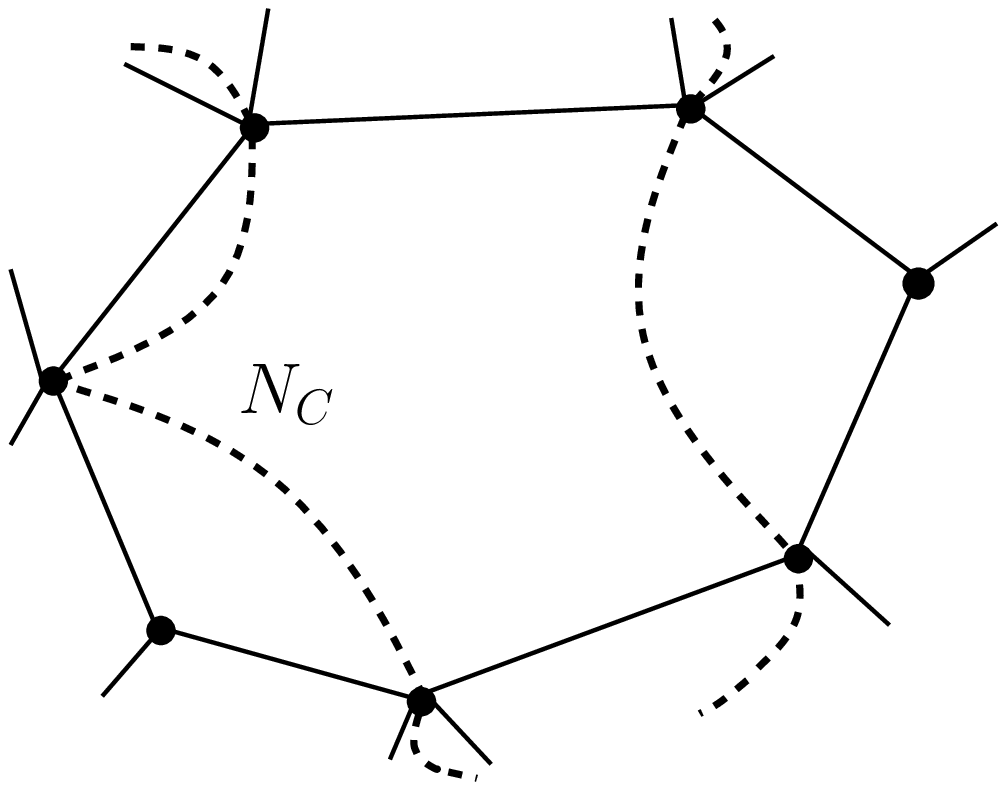}
 \includegraphics[width=0.33\textwidth]{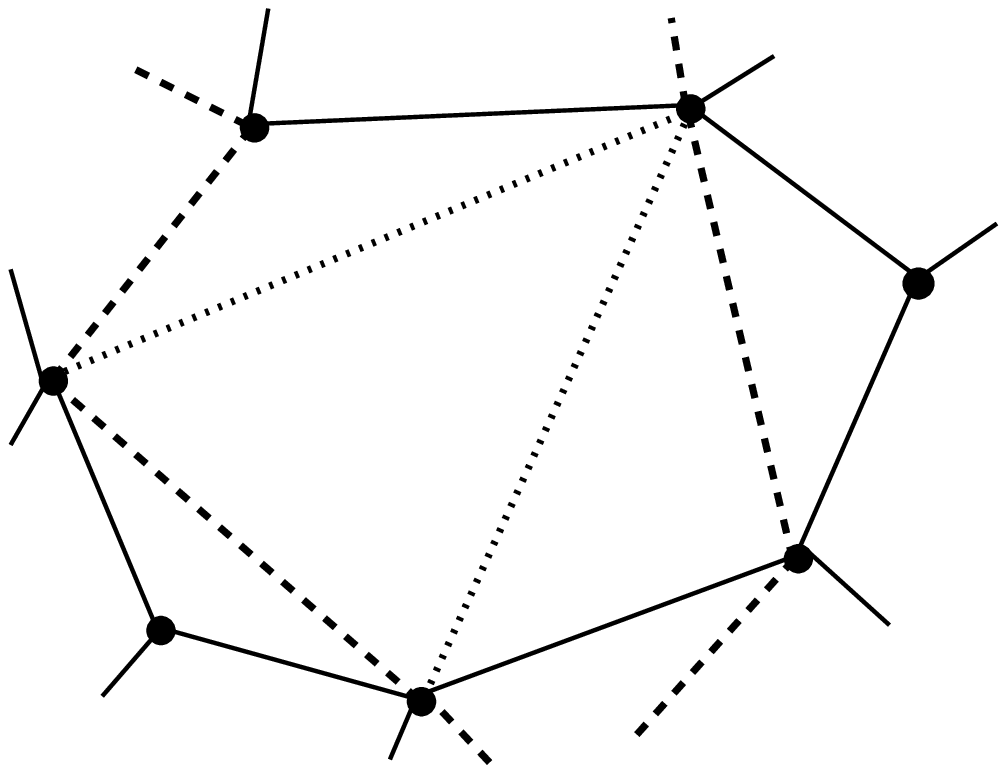}
\end{center}
\caption{{\small One the left, we see a face of our plane graph with a combinatorial noose $N_C$ indicated with dashed lines. To the right, we have mapped the noose to a cycle indicated with dashed lines and dotted lines indicating the face triangulation.}} 
\label{fig:noosecycle2}
\end{figure}
  \qed 

If $H$ is not 3-connected, a problem may occur in the last step of the previous proof when triangulating $H$.  Consider a  sub-sequence $[v_i,f_i,v_{i+1}]$ in $N_C$. We assume  there already exists an edge $e_i=\{v_i,v_{i+1}\}$ and  $v_i,v_{i+1}$ separate $H$, that is, $H$ is 2-connected. Then it may be the case that $e_i$ is not incident to $f_i$, and thus, any triangulation of $H$ has an edge crossing $N_C$. 
 We surpass this problem in the general case  by triangulating some auxiliary graph instead. For an edge $e=\{v,w\}$ of a graph $H$ we \emph{subdivide} $e$ by adding a vertex $u$ to $V(H)$ and replacing $e$ by two new edges $e_1=\{v,u\}$ and $e_2=\{u,w\}$. In a drawing of $H$, we 
  place point $u$ in the middle of the drawing of $e$ partitioning $e$ into
 $e_1$ and $e_2$. 

\begin{lemma}
\label{lem:walktriang_DAMN}
Let $H$ be plane  graph and $N_C= [v_0,f_0,v_1,f_1,\ldots,f_{k-1},v_k]$ 
 a combinatorial noose of $H$ with $|V(N_C)|>2$. Let $H^*$ be  obtained by subdividing every edge in $E(H)$.
  There exists a planar triangulation $H'$ of $H^*$  such that 
   $[v_0,v_1,\ldots,v_k]$ is a cycle in $H'$.   
\end{lemma}
\proof
The combinatorial  noose $N_C$ is a combinatorial noose in $H^*$, too. As for any two consecutive vertices $v_i,v_{i+1}$ in $N_C$ there is no edge in $H^*$ and each vertex in $N_C$ is unique, we may add edges to $H^*$ as in the proof of Lemma~\ref{lem:walktriang_NEW} and triangulate $H^*$. 
 \qed


\vspace{-2mm}\paragraph{\bf Proof of Lemma~\ref{lem:numradwalks}.} If $H$ is triangulated, we have with Lemma~\ref{lem:radialwalks_NEW} that every combinatorial noose corresponds to a unique cycle in $H$.
By Proposition~\ref{totallength}, the number of cycles in $H$ is bounded by $2^{1.53 k}$. Since for every edge of a cycle in $H$, we have two choices for 
 a combinatorial noose to visit an incident face, we get the overall upper bound of $2^{2.53 k}$ on the number of combinatorial nooses.
If $H$ is plane, we have to count the triangulations either of $H$ (Lemma~\ref{lem:walktriang_NEW}) or of $H^*$ (Lemma~\ref{lem:walktriang_DAMN}).
By Proposition~\ref{lem:tutte} and the comments below it, there are at most $2^{3.24 \ell}$ non-isomorphic  triangulations on $\ell$ vertices.
 Let us denote this set of 
 triangulated graphs by $\Phi$.
We note that $H$ (resp.~$H^*$) is a subgraph of some graph of $\Phi$, say of all graphs in $\Phi_H \subseteq \Phi$ with $|\Phi_H| \geq 1$.
 Since every triangulated graph is 3-connected, we have with Theorem~\ref{prop:whitney} that every graph $H'$ in $\Phi_H$ has a unique  embedding in $\Sigma$ up to homeomorphism.
 The plane graph $H$ (resp.~$H^*$) is then a subdrawing of a drawing equivalent to an arbitrary plane embedding of $H'$ in $\Sigma$. 
 Thus, the number of triangulations times the number of combinatorial nooses in each triangulation is an upper bound on the number of combinatorial nooses in $H$, here $2^{5.77 k}$ (resp.~in $H^*$, here $2^{9.77 k}$).
\qed

 For embedded dynamic programming on a sc-decomposition $\langle T,\mu, \pi \rangle$, we can  argue with Remark~\ref{radialcycle_NEW} that if $H$ is a subdrawing of $G$, then noose $N$ formed by the middle set $\mids(e)$ is a noose of $H$, too.
Recalling Remark~\ref{obs:comb_noose}, the alternating sequence of vertices and faces of $H$ visited by $N$ forms a  combinatorial noose $N_C$ in $H$.


 This observation allows us to discuss the results from a combinatorial point of view without the underlying topological arguments. Instead of nooses we will refer to combinatorial nooses in the remaining section.

\subsection{Embedded dynamic programming}
\label{subsec:embdp}
\noindent In embedded dynamic programming, the basic difference to usual dynamic programming is that we do not check for every partial solution for a given problem if or how it lies in the graph processed so far. Instead, we check how the graph that we have processed so far is intersecting  the entire solution, that is how the graph is \emph{embedded} into our solution. 
For subgraph isomorphism, we compute  every possible way the processed subdrawing $G_{sub}$ of $G$ is embedded in the plane pattern $H$ up to homeomorphism, subject to how the bound of $G_{sub}$ intersects $H$. This bound is a combinatorial noose $N$ separating $G_{sub}$ from the rest of $G$. The number of solutions we get is bounded by the number of combinatorial nooses in $H$ we can map $N$ onto.
 We describe the algorithm in what follows. 

\vspace{-2mm}\paragraph{\bf Dynamic programming.} We  root  sc-decomposition $\langle T,\mu, \pi \rangle$ at some node $r \in V(T)$. 
For each edge $e \in T$, let $L_e$ be the set of leaves of the subtree rooted at $e$. 
The subgraph $G_e$ of $G$ is induced by the edge set $\{\mu(v) \mid v \in L_e\}$. The vertices of $\mids(e)$ form a combinatorial noose $N$ that separates $G_e$ from the residual graph. 

Assuming $H$ is a subdrawing of $G$, the basic idea of embedded dynamic programming is that we are interested 
  in how the vertices of the combinatorial noose $N$ are intersecting faces and vertices of $H$.
  Since every noose in $G$ is a noose in $H$, we can map $N$ to a combinatorial noose $N^H$ of $H$,  bounding (clockwise) a unique subgraph $H_{sub}$ of $H$. 

In each step of the algorithm, all  solutions for a sub-problem in $G_e$ are computed, namely  all possibilities of how $N$ is mapped onto a combinatorial noose $N^H$ in $H$ that separates $H_{sub}$ from the rest of $H$, where   $H_{sub} \subseteq H$ is isomorphic to subgraphs of $G_e$.
For every middle set, we store this information in an array. It is updated in a bottom-up process starting at the leaves of $\langle T,\mu,\pi\rangle$. During this updating process it is guaranteed that the `local' solutions for each subgraph associated with a middle set of the sc-decomposition are combined into a `global' solution for the overall graph $G$.



\vspace{-2mm}\paragraph{\bf Step 0: Initializing the middle sets.}
 Let $G$ be a plane graph with a  rooted   sc-decomposition $\langle T,\mu, \pi \rangle$ and let $H$ be a plane pattern. 
 For every  middle set $\mids(e)$ of $\langle T,\mu, \pi \rangle$ let $N$ be the associated combinatorial noose in $G$ with face-vertex sequence of $F(N) \cup V(N)$. 
  Let $\mathfrak{L}$ denote the set of all combinatorial nooses of $H$ whose length is at most the length of $N$.
  We now want to map $N$ order preserving to each $N^H \in \mathfrak{L}$. We  map vertices of $N$ to both vertices and faces of $H$.
 Therefore, we consider partitions of $V(N)=V_1(N) \dot{\cup} V_2(N)$  where vertices in $V_1(N)$ are mapped to vertices of $V(H)$ and vertices in $V_2(N)$ to faces of $F(H)$. 
 We define a mapping $\gamma : V(N) \cup F(N)  \rightarrow V(H) \cup F(H)$
  relating $N$  to the combinatorial nooses in $\mathfrak{L}$.
 For every $N^H \in \mathfrak{L}$  on faces and vertices of set $F(N^H) \cup V(N^H)$ and for every partition $V_1(N) \dot{\cup} V_2(N)$ of $V(N)$ 
  mapping $\gamma$ 
  \emph{is valid} if
\squishlist 
\item[$a)$] $\gamma$ restricted to $V_1(N)$ is a bijection to $V(N^H)$;
\item[$b)$] for every $v \in V_1(N)$ we have  $\gamma(v) \in V(N^H)$, and for every $v \in V_2(N)$ we have $\gamma(v) \in F(N^H)$;
 \item[$c)$] for every $f \in F(N)$ we have $\gamma(f) \in F(N^H)$;
\item[$d)$] for every pair $v_h,  v_j \in V(N)$ such that $[\gamma(v_h),f ,\gamma(v_j)]$ is a subsequence of $N^H$ for a face $f \in F(N^H)$ 
  and for every vertex $v_i \in V(N)$ with $v_i$ lying inbetween $v_h$ and $v_j$ in the sequence $N$, we have $\gamma(v_i) = f$;
  \item[$e)$] for every $v_i \in V(N)$  and subsequence $[f_{i-1},v_i ,f_i]$  of 
  $N$, if 
  $\gamma(v_i) \in F(N^H)$, we have $\gamma(f_{i-1})=\gamma(v_i)=\gamma(f_i)$;
\item[$f)$] for every pair of vertices $w_i,w_j$ in $V(N^H)$: if $\{w_i,w_j\} \in E(H)$ then $\{\gamma^{-1}(w_i),\gamma^{-1}(w_j)\} \in E(G)$. 
\squishend 
 
 \noindent Items $a)$ to $c)$ say where to map the faces and vertices of $N$ to. Items $d)$ and $e)$ make sure that if two vertices $v_h,v_j$ in sequence $N = [\ldots,v_h,\underline{\ldots},v_j,\ldots]$ are mapped to two vertices $w_i,w_{i+1}$ that appear in sequence $N^H$ as $[\ldots,w_i,f_i,w_{i+1},\ldots]$ 
  then every face and vertex inbetween $v_h,v_j$  in sequence $N$ (here underlined) is mapped to face $f_i$. Item $f)$ rules out the invalid solutions, that is, we do not map a pair of vertices in $G$ that have no edge in common to the endpoints of an edge in $H$. We  do so because if $H$ is a subdrawing of $G$ then an edge in $H$ is an edge in $G$, too.   
For an illustration, see Figure~\ref{fig:dp}.
 \begin{figure}[h]
\begin{center}
\includegraphics[width=0.3\textwidth]{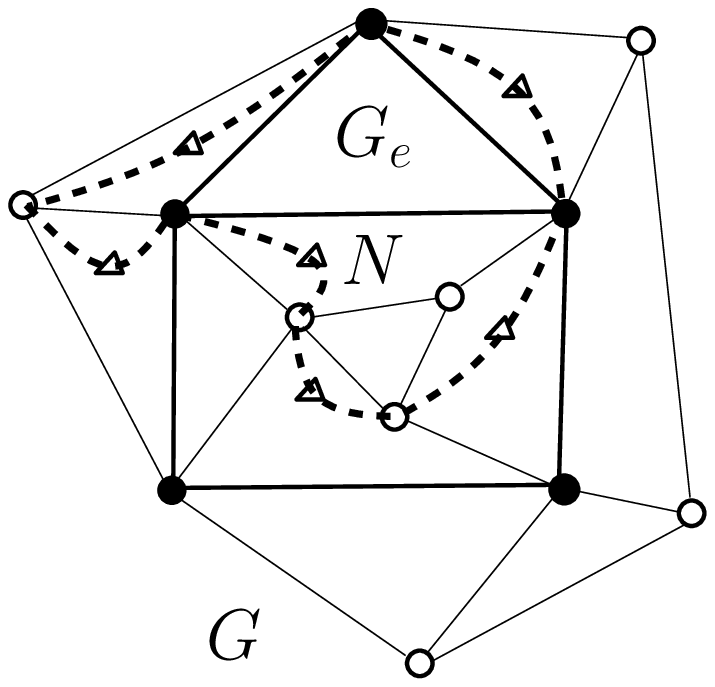}
 \includegraphics[width=0.3\textwidth]{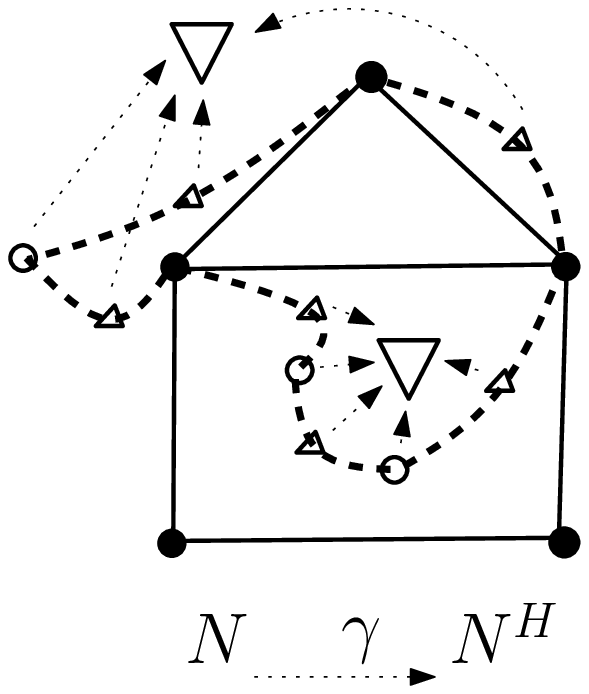}
 \includegraphics[width=0.18\textwidth]{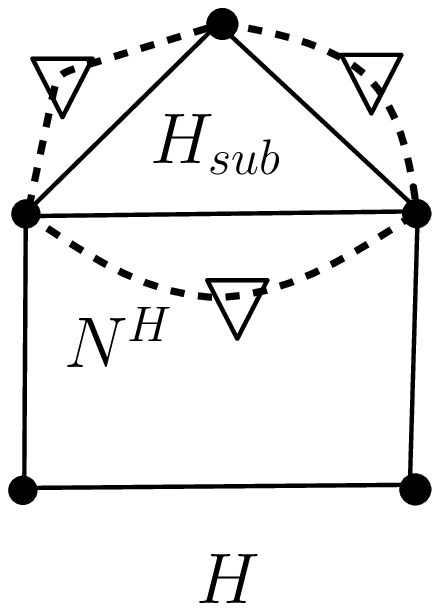}
\end{center}
\caption{{\small On the left, we have a plane graph $G$ with a subdrawing $H$ emphasized. A combinatorial noose $N$ separating subgraph $G_e$ is indicated by dashed lines. The vertices of $N$ are full and empty circles and the faces triangles.
  In the middle, we have $H$ and indicate to which faces (big triangles) of $H$ vertices and faces of $N$ are mapped by $\gamma$. This gives us combinatorial noose $N^H$ on the right, separating subgraph $H_{sub}$.}} 
\label{fig:dp}
\end{figure}  
 
 We assign an array $A_e$ to each $\mids(e)$ consisting of all tuples $\langle N^H,\gamma(N) \rangle$ each representing  a valid mapping $\gamma(N)$ from combinatorial noose $N$ corresponding to $\mids(e)$ to a combinatorial noose $N^H \in \mathfrak{L}$. 
 The vertices and faces of $N$ are oriented clockwise around $G_e$. 
  Without loss of generality, we assume for every $\langle N^H,\gamma(N) \rangle \in A_e$ the orientation of $N^H$ to be clockwise around the subgraph $H_{sub}$ of $H$ isomorphic to a subgraph of $G_e$.
 
\vspace{-2mm} \paragraph{\bf Step 1: Update process.} 
 We update the arrays of the middle sets in post-order manner from the leaves of $T$ to root $r$. 
  In each dynamic programming step, we compare the arrays of  two middle sets $\mids(e),\mids(f)$ in order to 
   create a new array assigned to the middle set $\mids(g)$, where $e,f$ and $g$ have a vertex of $T$ in common.
   From~\cite{DornPBF09} we know about a special property of sc-decompositions: namely that  the combinatorial noose $N_g$ is formed by the symmetric difference of the combinatorial nooses $N_e,N_f$ 
    and that $G_g= G_e \cup G_f$. In other words, we are ensured that if two solutions on  $G_e$ and $G_f$  bounded by $N_e$ and $N_f$ \emph{fit together}, then they form a new solution on $G_g$ bounded by $N_g$.
    We now determine when two solutions represented as tuples in the arrays $A_e$ and $A_f$ fit together. 
 We update two tuples $\langle N^H_e,\gamma_e(N_e) \rangle \in A_e$ and $\langle N^H_f,\gamma_f(N_f) \rangle \in A_f$ to a new tuple in $A_g$ if
 \squishlist
 \item for all $v \in V(N_e) \cap V(N_f)$,   $\gamma_e(v) = \gamma_f(v)$;
 \item for all $f \in F(N_e) \cap F(N_f)$,    $\gamma_e(f) = \gamma_f(f)$;
  \item for the subgraph $H_e$ of $H$ separated by $N^H_e$ and the  subgraph $H_f$ of $H$ separated by $N^H_f$, we have that 
   $E(H_e) \cap E(H_f) = \emptyset$  and
  $V(H_e) \cap V(H_f) \subseteq \{ \gamma(v) \mid  v \in  V(N_e) \cap V(N_f)\}$.
 \squishend
  If $N_e$ and $N_f$ fit together, we get a valid mapping $\gamma_g: N_g \rightarrow N^H_g$ as follows:
 \squishlist
 \item for every $x \in  ( \, V(N_e) \cup  F(N_e) \,) \cap ( \, V(N_f) \cup  F(N_f) \,) \cap  (\, V(N_g) \cup  F(N_g) ,)$ we have $\gamma_e(x) = \gamma_f(x) = \gamma_g(x)$; 
 \item for every $y \in  ( \, V(N_e) \cup  F(N_e) \,) \setminus ( \, V(N_f) \cup  F(N_f) \,)$ we have $\gamma_e(y)  = \gamma_g(y)$;
 \item for every $z \in  ( \, V(N_f) \cup  F(N_f) \,) \setminus ( \, V(N_e) \cup  F(N_e) \,)$ we have $\gamma_f(z)  = \gamma_g(z)$.
 \squishend
   We have that $\gamma_g$ is a valid mapping from $N_g$ to the combinatorial noose $N^H_g$ that bounds  subgraph $H_g = H_e \cup H_f$.
  Thus, we add  tuple $\langle N^H_g,\gamma_g(N_g) \rangle$ to  array $A_g$.
\vspace{-2mm}\paragraph{\bf Step 2: End of DP} If, at some step, we have a solution where the entire  subgraph $H$ is formed, we exit the algorithm confirming. That is, if $H = H_e \cup H_f$ and $H_i$ is bounded by $N_i$ (for both $i\in\{e,f\}$) 
then the combinatorial noose $N_g$  is bounding the subgraph of $G$ isomorphic to $H$.
  We are able to output this subgraph  by reconstructing the solution top-down in $\langle T,\mu, \pi \rangle$.
  If at root $r$ no subgraph isomorphic to $H$ has been found, we output 'FALSE'. 
  
\vspace{-2mm}\paragraph{\bf Correctness of DP} 
Let plane graph $H$ be a subdrawing of $G$.
We have seen already in Step 0 how we  map every combinatorial noose of $G$ that identifies a separation of $G$ via a valid mapping $\gamma$ to a combinatorial noose of $H$ determining a separation of $H$.
 Every edge of $H$ is bounded by a combinatorial noose $N^H$ of length two, which is determined by tuple $\langle N^H,\gamma(N) \rangle$ in an array assigned to a leaf edge of $T$.
We need to show that Step 1 computes a valid solution for $N_g$ from $N_e$ and $N_f$ for incident edges $e,f,g$.
We note that the property that the symmetric difference of the combinatorial nooses $N_e$ and $N_f$ forms a new combinatorial noose $N_g$ is passed on to the combinatorial nooses  $N^H_e,N^H_f$ and   $N^H_g$ of $H$, too.
If the two solutions fit together, then $H_e$ of $H$ separated by $N^H_e$ and subgraph $H_f$ of $H$ separated by $N^H_f$ only intersect in the image of $V(N_e) \cap V(N_f)$. We may observe that  $N^H_e$ and $N^H_f$ intersect in a continuous alternating subsequence  with order reversed to each other, i.e., $N^H_e\mid_{N_e \cap N_f} = \overline{N^H_f}\mid_{N_e \cap N_f}$, where $\overline{N^H}$ means the reversed sequence $N^H$.
Since every oriented $N^H$ identifies uniquely a separation of $E(H)$, we can easily determine if two tuples $\langle N^H_e,\gamma_e(N_e) \rangle \in A_e$ and $\langle N^H_f,\gamma_f(N_f) \rangle \in A_f$  fit together and form a new subgraph of $H$. If $H$ is a subdrawing of $G$, then at some step we will enter Step 2 and produce the entire $H$.

\vspace{-2mm}\paragraph{\bf Running time analysis.}  
We first give an upper bound on the size of each array. The number of combinatorial nooses in $\mathfrak{L}$ we are considering is bounded by  the total number of combinatorial nooses in $H$, which is $2^{O(|V(H)|)}$ by Lemma~\ref{lem:numradwalks}.
 The number of partitions of vertices of any combinatorial noose $N$ is bounded by $2^{|V(N)|}$. Since the order of both $N^H$ and $N$ is given 
  we only have $2 |V(H)|$ possibilities to map vertices of $N$ to $N^H$, once the vertices of $N$ are partitioned. 
   Thus, in an array $A_e$ we may have up to $2^{O(|V(H)|)} \cdot 2^{|V(N)|} \cdot |V(H)|$ tuples $\langle N^H_e,\gamma(N_e) \rangle$. 
   We first create all tuples in the arrays assigned to the leaves. Since middle sets of leaves only consist of an edge in $G$, we get arrays of size $O(|V(H)|^2)$ which we compute in the same asymptotic running time.
   When updating middle sets $\mids(e),\mids(f)$, we  compare every tuple of one array $A_e$ to every tuple in array $A_f$ to check if two tuples fit together. We can compute the unique subgraph $H_e$ (resp.~$H_f$) described by a tuple in $A_e$ (resp.~$A_f$), compare two tuples in $A_e,A_f$ and create a new tuple in $A_g$ in time linear in the order of $V(N)$ and $V(H)$. Since the size of $A_g$ is bounded by $2^{O(|V(H)|)} \cdot 2^{O(|V(N)|)}$, the update process for two middle sets takes the same asymptotic time.
    Assuming  sc-decomposition $\langle T,\mu,\pi\rangle$ of $G$ has  width $\omega$ and $|V(H)| \leq \omega$, we get the following result.

\begin{lemma}
 For a plane graph $G$ with a given sc-decomposition $\langle T,\mu,\pi\rangle$ of $G$ of width $w$ and a plane pattern $H$ on $k \leq w$ vertices we can search for a subdrawing of $G$ equivalent to $H$ in time $2^{O(w)} \cdot n$.
\end{lemma}
 


\subsection{The algorithm}
\label{subsec:alg}
\noindent We present the overall algorithm for solving {\sc Plane Subgraph Isomorphism} with running time stated in Theorem~\ref{theor:planeSI}.

\begin{algorithm}[h]
    \dontprintsemicolon
    \Input{Plane graph $G$; Plane pattern $H$ of order $k$.}
    \BlankLine\;
    \lnl{every_v}
    Choose an arbitrary vertex $v$ in $G$.\;
  \lnl{ln:partition}
   Partition $V(G)$ into $S_0 \cup S_1 \cup \ldots \cup S_{\ell}$ with $S_i = \{w \in V(G): \mbox{dist}(v,w)=i\}$\; 
   \lnl{ln:subgraph}
   \For{every $G_i= G[S_i\cup\ldots\cup S_{i+k}]$ with $0 \leq i \leq \ell - k$}{ 
 \lnl{ln:sc-decomp}
  Compute sc-decomposition $\langle T,\mu, \pi \rangle$ of $G_i$.\; 
   \lnl{ln-dp}
   Do embedded dynamic programming on $\langle T,\mu, \pi \rangle$  to find a subgraph of $G_i$ isomorphic to $H$ and intersecting $S_i$.\;}
    \caption{Plane subgraph isomorphism: PLSI.}
    \label{alg:dp}
\end{algorithm}


Partitioning the vertex set in Line~\ref{ln:partition} of Algorithm~\ref{alg:dp} PLSI, is a similar approach to the well-known Baker-approach~\cite{Baker94}. Every vertex set $S_i$ contains the vertices of distance $i$ to the chosen vertex $v$.
$S_0=\{v\}$ and $\ell$ is the maximum distance in $G$ from $v$. The graph $G_i$ in Line~\ref{ln:subgraph} is induced by the sets $S_i,\ldots,S_{i+k}$.  
As in~\cite{Eppstein99}, we may argue that every vertex in $G$ appears in at most $k$ subgraphs $G_i$. 
This keeps our running time linear in $n$.
We can apply  Lemma~\ref{lem:tamaki}  to each $G_i$ in Line~\ref{ln:sc-decomp} to a compute sc-decomposition $\langle T,\mu, \pi \rangle$  of width $\leq 2k$, by adding a root vertex $r$ for the BFS tree and make $r$ adjacent to every vertex in $S_i$.
 The dynamic programming approach  can easily be turned into an algorithm  counting subgraph isomorphisms (similar to~\cite{Eppstein99}), by
  using a counter in the dynamic programming. Using an inductive argument, for every subgraph $G_i$ in Line~\ref{ln-dp} we only compute subgraphs intersecting with vertices in $S_i$ and thus omit double-counting. We can also adopt our technique to list the subgraphs of $G$ isomorphic to $H$.

\section{Planar subgraph isomorphism}
\label{sec:planarsubiso}

\noindent Now we consider the case when both pattern $H$ and host graph $G$ are planar but not embedded.
  However, we observe that if $H$ is isomorphic to a subgraph of $G$, then for every planar embedding of $G$ there exists a drawing of $H$ that is equivalent to a subdrawing of $G$. Hence, we may simply embed $G$ planarly, and run the algorithm of the previous section for all  non-equivalent embeddings of $H$.
  The following lemma tells us that the number of times we call the algorithm is restricted, too.  

\begin{lemma}
 \label{lemma:equemb}
Every planar $k$-vertex graph  has $2^{O(k)}$ non-equivalent embeddings in $\Sigma$. 
\end{lemma} 
\proof By Proposition~\ref{lem:tutte}, there are at most $2^{3.24 k}$ non-isomorphic maximal planar graphs on $k$ vertices.
 Every planar graph $H$ is a subgraph of a maximal planar graph.
 Every maximal planar graph has a unique embedding which is a triangulation.
  Thus, every embedding of $H$ is a subdrawing of a triangulation of $H$. 
The number of such subgraphs is bounded by the number of edge subsets of $H'$, since for every edge subset of $S \subseteq E(H')$ of same cardinality as $E(H)$, $H'[S]$ may be isomorphic to $H$.  In this case, $H'[S]$ then gives a possible  embedding of $H$ in $\Sigma$. 
Hence, the number of  embeddings of $H$ in $\Sigma$ up to homeomorphism is bounded by $2^{6.24 k}$.
\qed

\vspace{-2mm} \paragraph{\bf The whole algorithm}

\begin{algorithm}[h]
    \dontprintsemicolon
    \Input{Planar graph $G$, Planar pattern $H$ of size $k$.}
    \BlankLine\;
    Compute a planar embedding of $G$.\;
     \lIf{$H$ triangulated or 3-connected} {Return PLSI$(G,H)$.}\;
\For{every non-equivalent embedding $I$ of $H$}{ 
 Return PLSI$(G,I)$.}\;
    \caption{Planar subgraph isomorphism.}
    \label{alg:main}
\end{algorithm}

We compute in Algorithm~\ref{alg:main} every non-equivalent embedding of $H$ using the constructive proof of Lemma~\ref{lemma:equemb}. 
That is, we compute the set $\mathcal{H}$ of non-isomorphic maximal planar graphs in time proportional to its size using the algorithm in~\cite{LiN01}. 
  For every graph $H' \in \mathcal{H}$ and every subdrawing $I$ of $H'$ we check whether  $I$ is isomorphic to $H$ by using the linear time algorithm for planar graph isomorphism in~\cite{HopcroftW74}\footnote{We get a list of embeddings of $H$, from which we can delete equivalent drawings by a modification of the algorithm in~\cite{HopcroftW74}---namely isomorphism test for face-vertex graphs.}  .
  By Lemma~\ref{lemma:equemb}, we then call  Algorithm~\ref{alg:dp}   $2^{O(k)}$ times, for each plane drawing $I$ isomorphic to $H$.  This ensures us that Algorithm~\ref{alg:main} has running time as stated in Theorem~\ref{bigtheorem}. 


\section{Conclusion}

\noindent We have shown how to use topological graph theory to improve the results on the already mentioned variations of {\sc Planar Subgraph Isomorphism}, solving the open problems posed in~\cite{Eppstein99} and~\cite{DornPBF09}. 
With the results of~\cite{Eppstein00},~\cite{Eppstein99} extends the feasible graph class from planar graphs to apex-minor-free graphs. 
 This cannot be done with the tools presented here. However,~\cite{DornFT07} devise a truly subexponential algorithm for {\sc $k$-Longest Path} in $H$-minor-free graphs and thus apex-minor-free graphs, employing the structural theorem of Robertson and Seymour~\cite{RobertsonS03} and the results of~\cite{DemaineHK05,DawarGK2007,DemaineHT05}. Can the structure of $H$-minor-free graphs, be exploited for our purposes?  

It seems unlikely that our work can be extended to obtain a subexponential algorithm. 
The first reason, mentioned in the introduction, is that Bidimensionality applies to subgraphs with minor properties rather than to general subgraphs. Secondly,  our enumerative bounds are either tight or of lower bound $2^{\Omega(k)}$. We want to pose the open problem: Is  {\sc Plane Subgraph Isomorphism}   solvable in time~$2^{o(k)}n^{O(1)}$?

\medskip
 
 \noindent{\bf Acknowledgments}. 
The author thanks Paul Bonsma, Holger Dell and Fedor Fomin for discussions and comments of great value to the presentation of these results.


\appendix

\section{SC-Decompositions in linear time}
\label{app:tamaki_alg}

For a plane graph $G$ we define a \emph{radial
graph}~$R_G$ as follows: $R_G$ is a bipartite graph with the
bipartition $F(G) \cup V(G)$. A vertex $v\in V(G)$ is adjacent in
$R_G$ to a vertex $f\in F(G)$ if and only if~the vertex $v$ is
incident to the face $f$ in the drawing of~$G$.

Let $G$ be a plane graph with some vertex $r \in V(G)$  and $R_G$ its radial graph. 
Let $T$ be a spanning tree of $G$ rooted at $f$ that is determined by breadth first search.
Choose a face $f$ adjacent to root $r$. 
If the longest path from $r$ to a leaf of $T$ is $\ell$ then the distance $d_f$ in the radial graph $R_G$ from vertex $f$ to any other (face)vertex $x$ is at most $2 \ell +1$. This is due to the fact that  there exists an edge $\{f,r\}$ in $R_G$, and for every edge in $T$ there is a detour in $R_G$ of at most two edges.  
  \cite{Tamaki03} show how to obtain a branch decomposition of width $d_f$ out of a BFS spanning tree rooted at $r$ of the radial graph\footnote{In fact the authors construct a carving decomposition out of the spanning tree of the dual graph of the radial graph that one obtains after deleting the dual edges of the BFS spanning tree.}. 
  Set $f$ to be the outer face of $G$. Let $T$ be a BFS spanning tree of   
 $R_G$ rooted at $f$ and let $\ell$ be the maximum distance in $T$ from $r$ to a leaf. 
We give now a compact presentation of the algorithm of~\cite{Tamaki03} and show that it translates to constructing a sphere-cut decomposition of $G$.
We define \emph{contracting} a vertex $v$ as identifying all vertices of $N(v)$ to a single vertex and deleting $v$.

\begin{algorithm}[h]
    \dontprintsemicolon
    \Input{ Plane  graph $G$, face $f \in F(G)$ , radial graph $R_G$ .}
    \Output{Branch-decomposition of $G$ of width at most $2 \ell +1$.}
    \BlankLine\;
    \lnl{BFS}
     Construct embedded BFS tree $T_S$ of $R_G$ at root $f$.\;
     \lnl{sc_line1}
      Set $T^* = R_G^* \setminus E(T_S)^*$ the dual graph of $R_G$ without the edge set dual to $T_S$\;
       \lnl{sc_line2}
       \For{every node $v$ in $T^*$}{
        \lnl{sc_line3}
       \lIf{$\deg(v)_T= 1$}{\;
        \lnl{sc_line4}
        \quad create $C_v$  a single edge, with nodes labeled $\{v\}$ and $\{N(v)\}$;}\;
     \lnl{sc_line5}   \lElse \;
      \lnl{sc_line6}  \quad create embedded ternary tree $C_v$ with $|N(v)|+1$ leaves;\;
        \lnl{sc_line7} \quad label one leaf with $\{v\}$ and the other leaves with $N(v)$ keeping a clockwise order.\; }
         
        \lnl{sc_line8}  (\emph{in post order}) \For{every edge $\{u,v\}$ in $T^*$, where $v$ is the parent node}{
        \lnl{sc_line10} \quad  combine $C_u$ and $C_v$ by identifying leaf $\{v\}$ in $C_u$ with leaf $\{u\}$ in $C_v$, and\;  
          \lnl{sc_line11} \quad  contract  the identified node and set new tree to be named $C_v$.}\;
            \lnl{sc_line12} Return ($C_r$ (for $r$ root of $T^*$)).\;   
     
    \caption{Computing SC-decomposition.}
    \label{alg:alg_tamaki}
\end{algorithm}

In~\cite{Tamaki03}, Algorithm~\ref{alg:alg_tamaki} is  proved  to compute a branch decomposition of planar graph $G$ of width $2 \ell +1$ in time $O( \ell n)$. 
 \begin{claim}
 Algorithm~\ref{alg:alg_tamaki} computes a sc-decomposition of $G$.
 \end{claim}
\begin{proof}
Observe that the edge set $T^*$ of the dual graph of $R_G$ (the so-called medial graph) minus the dual edges of $T_S$ in Line~\ref{sc_line1} forms a tree due to the acyclicity of $T_S$. Every node of $T^*$ corresponds to an edge of $G$  and in fact, spans the edge set of $G$. 
 For turning $T^*$ into a branch-decomposition we  $a)$  bijectively map the leafs of $T^*$ to the edges of $G$ and $b)$ make $T^*$ ternary. For $a)$ we generate for every node $v$  in $T^*$  from Line~\ref{sc_line2}--\ref{sc_line7} one ternary tree (or single edge tree), a \emph{local} tree $C_v$ with $v$ one leaf.  In Line~\ref{sc_line8}--\ref{sc_line12} those local trees are merged from the leaves of $T^*$ to its root in post order such that each $C_v$ contributes to leaf $v$ in the overall ternary tree $C_r$. 
 We show now that the such formed branch decomposition actually obeys our definition of an sc-decomposition, that is the vertices in each middle set form a cycle in the radial graph. Note that every edge of $G$ forms a $4$-cycle in $R_G$. Let $e^* \in E(T^*)$ be the dual edge of edge $e= \{f,g\} \in E(R_G \setminus T_S)$. Then the union of $e$ and  the path through $T_S$ from $f$ over the lowest common ancestor of $f,g$ in $T_S$ to $g$ forms a cycle in $R_G$ that separates the two subtrees of $T^*$ that are separated by $e^*$. Thus, $T^*$ already possesses middle sets that form  nooses in $G$. However $T^*$ is not ternary since it may have maximum node degree $4$. The leaves of each local tree $C_v$ in Line~\ref{sc_line6} are embedded in  the same order as the inverse of their labels, the neighboring nodes of $v$ in $T^*$, 
 and thus we keep the same ordering in the overall ternary tree $C_r$. 
  Every edge $e$ of $C_r$ comes from an edge of one of the local trees $C_v$, and $e$ separates the neighbors $N_1(v)$  from $N_2(v)$ where the disjoint union $N_1(v),N_2(v)$ form the neighborhood $N(v)$ of $v$ in $T^*$. Like this, $T^*$ falls apart into two subtrees each bounded by a cycle in $R_G$ formed similarly as above by the union of $N_i(v)$ bounding minimal path in  $T_S$ and the path through the edges of $T^*$ induced by $N_i[v]$, for  $i=1$ and $i=2$ respectively.


\end{proof}




\end{document}